\documentclass[10pt,onecolumn,twoside]{article}

\usepackage{amsmath,amssymb,amsfonts,amsthm,dsfont}
\usepackage{verbatim}

\usepackage{algorithm}
\usepackage{algorithmic}
\usepackage{url}

\usepackage[dvips]{graphicx}
\usepackage[dvipsnames]{xcolor}

\newcommand{\isdef}{\stackrel{\textrm{def}}{=}}

\DeclareMathOperator*{\esssup}{ess\,sup}
\allowdisplaybreaks %for align environment and others ?

\theoremstyle{plain} 
\theoremstyle{remark} \newtheorem{remark}{\textbf{Remark}}
\theoremstyle{plain} \newtheorem{theorem}{\textbf{Theorem}}
\theoremstyle{plain} 
\theoremstyle{plain} 
\theoremstyle{definition} 
\theoremstyle{plain}

\makeatletter
\newcommand{\pushright}[1]{\ifmeasuring@#1\else\omit\hfill$\displaystyle#1$\fi\ignorespaces}
\newcommand{\pushleft}[1]{\ifmeasuring@#1\else\omit$\displaystyle#1$\hfill\fi\ignorespaces}
\makeatother

\newcommand*\rfrac[2]{{}^{#1}\!/_{#2}}
\newcommand{\squash}[1]{\raisebox{0.04ex}[0pt][0pt]{\small$\textstyle #1$}}

\begin{document}

\title{The expected bit complexity of the von Neumann rejection algorithm}
\author{Luc Devroye\footnote{School of Computer Science, McGill University, Canada}\,~and~Claude Gravel\footnote{Department of Computer Science and Operations Research, Universit\'e de Montr\'eal, Canada}}
\date{\today}
\maketitle

\begin{abstract}
In 1952, von Neumann introduced the rejection method
for random variate generation. We revisit this algorithm
when we have a source of perfect bits at our disposal.
In this random bit model, there are universal lower bounds for
generating a random variate with a given density to
within an accuracy $\epsilon$ derived by Knuth and Yao, and
refined by the authors. In general, von Neumann's method
fails in this model. We propose a modification that insures
proper behavior for all Riemann-integrable densities on
compact sets, and show that the expected number of
random bits needed behaves optimally with respect to universal lower bounds. In particular, we introduce the notion of an oracle that evaluates the supremum and infimum of a function on any rectangle of $\mathbb{R}^{d}$, and develop a quadtree-style extension of the classical rejection method.



\end{abstract}

\section{Introduction}

\textbf{\Large{T}}\normalsize{}he purpose of this paper is to discuss von Neumann's \cite{vN51} rejection method to generate a random variable $X$ under the random bit model. In this model, we have access to an infinite sequence of independent random Bernoulli$(\rfrac{1}{2})$ bits, and are interested in the complexity as measured by the number of bits used until halting. For integer-value random variables $X$, the entire story is known. Knuth and Yao \cite{KnuthYao1976} obtained lower bounds on the expected number of bits for the exact simulation of $X$, and exhibited optimal algorithms. On the other hand, for random variables $X$ with a density on $\mathbb{R}^{d}$, no exact algorithm exists, since any output of an algorithm delivers a function of a (possibly random) number of random bits. Thus, it is necessary to introduce the notion of an $\epsilon$-approximation (see section \ref{sect:bisection}) using Wasserstein $L_{\infty}$-distance (Devroye and Gravel \cite{DevGra2015paper1} and Rachev and R\"{u}schendorf \cite{RatRus1998}).

A naive application of the rejection method---one of the most often used methods in simulation---leads to errors and inconsistencies. To deal with this, we introduce the notion of an oracle that computes the supremum and infimum of a function over any rectangles of $\mathbb{R}^{d}$ (see section \ref{sect:compactcase}). The oracle can be used in conjunction with a quadtree partition of the space to design a rejection algorithm that is guaranteed to deliver an $\epsilon$-approximation (sections \ref{sect:compactcase} and \ref{sect:noncompacta1}). We show that is valid whenever $f$ is Riemann-integrable and derive expected complexity bounds (in terms of the number of random bits consumed) that are close to the universal lower bounds of \cite{DevGra2015paper1}.

We believe that random number generation libraries should offer the possibility of specifying any $\epsilon$---no matter how small---as an input. Current practice entirely disregards the effects of approximations. In this paper and a companion paper \cite{DevGra2015paper1}, we take a first small step towards this goal. Applications in physics (see \cite{Karney_2016_normal_sampling} who requires and develops $\epsilon$-accurate normal generators), quantum computing (see Brassard, Devroye, and Gravel \cite{ghz_ieee_BrDeGr_2016}), and other areas of science demand very precise computations. Much more is needed of course, especially when many random variables are combined in scientific computation---how does an $\epsilon$-approximation propagate through a system, for example?

\section{The discrete case}\label{sect:discrete}

Consider two probability vectors $(p_{i})_{i\in\mathbb{Z}}$, and $(q_{i})_{i\in\mathbb{Z}}$, where
\begin{displaymath}
\sup_{i\in\mathbb{Z}}\frac{p_{i}}{q_{i}}\leq C
\end{displaymath}
for a finite constant $C$. Then von Neumann's rejection method can be reformulated as follows, if $C$ is explicitly known:
\begin{flushright}
\color{gray}{\textbf{The algorithm is at the beginning of the next page.}}\color{black}{}
\end{flushright}
\vfill
\clearpage

\begin{algorithm}[h]
\caption{Von Neumann's method for discrete distributions in the random bit model}
\begin{algorithmic}[1]\label{algo:orinatumsupolov}
\LOOP
\STATE Generate $X$ according to $(q_1,q_2,\ldots)$. \COMMENT{In the random bit model, the Knuth-Yao or Han-Hoshi algorithm can be used here to generate $X$.}
\STATE Generate $U$ uniformly on $[0,1]$.
\IF{ $UCq_{X}\leq p_{X}$}
\RETURN $X$ \COMMENT{Exit}
\ENDIF
\ENDLOOP
\end{algorithmic}
\end{algorithm}

The expected number of iterations in this algorithm is $C$. The returned random variable, $X$, has distribution $(p_{i})_{i\in\mathbb{Z}}$. Knuth and Yao \cite{KnuthYao1976} showed that to generate $X$, the expected number of random bits required by any algorithm is at least the entropy of $X$,
\begin{displaymath}
\mathcal{E}(X)\isdef\sum_{i\in\mathbb{Z}}{q_i\log_{2}\bigg(\frac{1}{q_i}\bigg)}.
\end{displaymath}
They also exhibited an algorithm that requires an expected number of bits not exceeding $\mathcal{E}(X)+2$. Note that given $X$, the decision
\begin{displaymath}
UCq_{X}\leq p_{X}
\end{displaymath}
can be made using $2$ expected random bits because to decide $U\leq p_{X}\slash Cq_{X}$ given $X$ follows a geometric law with parameter $\rfrac{1}{2}$; we compare the $i^{\text{th}}$ bit of $U$ with the $i^{\text{th}}$ bit of $p_{X}\slash Cq_{X}$ until both don't agree for integers $i>0$ (see, e.g., Devroye and Gravel \cite{DevGra2015paper1}). The expected number of random bits needed by this implementation is not more than
\begin{displaymath}
C\big(\mathcal{E}(X)+2\big).
\end{displaymath}
This is usually quite far from the lower bound of Knuth and Yao, $\sum_{i\in\mathbb{Z}}{p_{i}\log_{2}\frac{1}{p_i}}$. It is worth to mention an application of the rejection method in the bit model to the simulation of physical phenomena and to communication complexity in Brassard, Devroye, and Gravel \cite{ghz_ieee_BrDeGr_2016}. The remainder of the paper is concerned with the case in which $X$ has a density $f$. Since such $X$ cannot be generated exactly by any algorithm, it has to be approximated in some manner by a discrete random variable, and thus we require an appropriate---and as it turns out, nontrivial---generalization of Algorithm \ref{algo:orinatumsupolov}.

\section{Bisection algorithm}\label{sect:bisection}

The building block for continuous random generation is the bisection algorithm, which is mathematically equivalent to an algorithm given in Devroye and Gravel \cite{DevGra2015paper1}. We develop a version here that is convenient for later use. The analysis of Theorem \ref{thm_bisect} below is new. Consider an algorithm for generating a random variable $X$ with density $f$ on $\mathbb{R}^{d}$ that has the following property: for a given $\epsilon>0$, it outputs a random variable $X_{\epsilon}\in\mathbb{R}^{d}$---an $\epsilon$-approximation of $X$---such that there exists a coupling of $(X,X_{\epsilon})$ with
\begin{displaymath}
\esssup\|X-X_{\epsilon}\|\leq\epsilon,
\end{displaymath}
where $\|\cdot\|$ is the $L_{\infty}$-distance in $\mathbb{R}^{d}$. It is understood that necessarily, $X_{\epsilon}$ is discret since it is a function of a (random) finite number of random bits. We call $X_{\epsilon}$ an $\epsilon$-approximation of $X$.

Devroye and Gravel \cite{DevGra2015paper1} showed that if $T_{\epsilon}$ is the random number of bits needed by any algorithm for generating such an approximation $X_{\epsilon}$, then
\begin{align}
\mathbf{E}(T_{\epsilon})\geq \mathcal{E}(f)+d\log_{2}\bigg(\frac{1}{\epsilon}\bigg)-d,\label{lower_bnd}
\end{align}
where $\mathcal{E}(f)$ is the differential entropy of $f$,
\begin{displaymath}
\mathcal{E}(f)=\int{f\log_{2}\bigg(\frac{1}{f}\bigg)}.
\end{displaymath}
The lower bound is valid under a technical condition known as R\'{e}nyi's condition (see R\'{e}nyi \cite{Renyi1959} and Csisz\`{a}r \cite{Csiszar1961}), namely that the entropy of the discrete random variable $\lfloor X\rfloor$, which takes values on the grid of all integer-valued vectors of $\mathbb{R}^{d}$, be finite. The lower bound (\ref{lower_bnd}) serves as a guide to calibrate and compare the rejection algorithms presented in this paper. It is especially crucial to match its main term, \mbox{$d\log_{2}\big(\squash{\frac{1}{\epsilon}}\big)$}, without an extra multiplicative constant.

We require an auxiliary set of results on bisection algorithms for generating a random variate that is an $\epsilon$-approximation of a random variable $X$ with a continuous (not necessarily absolutely continuous) distribution function $G$ on a compact interval $[a,b]$ of length $L\isdef|b-a|$. The bisection Algorithm \ref{bisecto_algo} assumes that we have access to both $G$ and $G^{-1}$.

\begin{flushright}
\color{gray}{\textbf{The algorithm is at the beginning of the next page.}}\color{black}{}
\end{flushright}
\vfill
\clearpage

\begin{algorithm}
\caption{The bisection algorithm in the random bit model (Devroye and Gravel \cite{DevGra2015paper1})}\label{bisecto_algo}
\begin{algorithmic}[1]
\STATE $J\leftarrow [G(a),G(b)]=[0,1]$
\LOOP
\IF {$|I|\isdef b-a\leq 2\epsilon$}
\STATE $X_{\epsilon}\leftarrow \frac{a+b}{2}$
\RETURN $X_{\epsilon}$ \COMMENT{Exit}
\ELSE
\STATE $B\leftarrow \texttt{random unbiased bit}$.
\STATE $z\leftarrow G^{-1}\Big(\frac{G(a)+G(b)}{2}\Big)$.
\IF {$B=0$}
\STATE $I\leftarrow[a,z]$
\STATE $J\leftarrow\Big[G(a),\frac{G(a)+G(b)}{2}\Big]=\big[G(a),G(z)\big]$
\ELSE%[$B=1$]
\STATE $I\leftarrow[z,b]$
\STATE $J\leftarrow\Big[\frac{G(a)+G(b)}{2}, G(b)\Big]=\big[G(z),G(b)\big]$
\ENDIF
\ENDIF
\ENDLOOP
\end{algorithmic}
\end{algorithm}

\begin{theorem}\label{thm_bisect}
For the bisection Algorithm \ref{bisecto_algo} applied to any distribution with support on an interval of length $L$, and halted as soon as an interval of length less than or equal to $2\epsilon$ is reached, we have
\begin{enumerate}
\item[(i)] If $X_{\epsilon}$ denotes the center of the halting interval, then there exists a coupling between $X$ and $X_{\epsilon}$ such that $\|X-X_{\epsilon}\|\leq \epsilon$.
\item[(ii)] If $T_{\epsilon}$ denotes the number of bits used, then
\begin{displaymath}
\mathbf{E}(T_{\epsilon})\leq 3+\log_{2}^{+}\bigg(\frac{L}{2\epsilon}\bigg),
\end{displaymath}
where $\log_{2}^{+}(x)=\max\{0,\log_{2}(x)\}$.
\end{enumerate}
\end{theorem}

\begin{remark}
We note here that in general this bound cannot be improved by more than $3$ bits. Just consider the uniform distribution on $[0,L]$. Since all intervals have length to $\frac{L}{2^{i}}$ after $i$ were used, we have
\begin{displaymath}
T_{\epsilon}=\min\bigg\{i\geq 0:\phantom{1}\frac{L}{2^{i}}\leq 2\epsilon\bigg\}=\max\bigg\{0,\bigg\lceil\log_{2}\frac{L}{2\epsilon}\bigg\rceil\bigg\}.
\end{displaymath}
\end{remark}

\begin{proof}[Proof of Theorem \ref{thm_bisect}]
The bisection method yields, very naturally, a full binary tree, i.e., one in which all internal nodes have two children. Each internal node corresponds to a subinterval of $[a,b]$ of length greater than $2\epsilon$, the root represents the original interval $[a,b]$ of length $L$, and leaves represent intervals of length less than or equal to $2\epsilon$ that cause an exit.

Upon exit, the random variable $X_{\epsilon}$ can be coupled with $X$ such that $\|X-X_{\epsilon}\|\leq \epsilon$, because at every iteration, the random binary choice picks the correct interval, $[a,z]$ or $[z,b]$, with the correct probability $\rfrac{1}{2}$. One could thus define $X$ as the limit of $I$ when the algorithm is run without halting. Since $X_{\epsilon}$ is the midpoint of an interval of length at most $2\epsilon$ that also contains $X$, we must have $\|X-X_{\epsilon}\|\leq\epsilon$. This shows part (i).

To prove part (ii), let us denote the set of leaves by $\mathcal{L}$, and the set of internal nodes (i.e., all non-leaf nodes) by $\mathcal{I}$. The depth of node $u$ is denoted by $\mathrm{d}(u)$. It is of course possible that $\mathcal{I}$ and $\mathcal{L}$ are both infinite. However, one has that in all cases,
\begin{displaymath}
\sum_{u\in\mathcal{L}}{\frac{1}{2^{\mathrm{d}(u)}}}\leq 1,
\end{displaymath}
because the leaves form a non-overlapping covering of $[a,b]$ so that the bisection method---a random walk started at the root and halted when a leaf is reached---always stops. Also,
\begin{displaymath}
\mathbf{E}(T_{\epsilon})=\sum_{u\in\mathcal{L}}{\frac{\mathrm{d}(u)}{2^{\mathrm{d}(u)}}}.
\end{displaymath}
In the next chain of inequalities, we define
\begin{displaymath}
N_{\ell}=\sum_{v\in\mathcal{I}}{\mathds{1}_{\{\mathrm{d}(v)=\ell\}}},\phantom{1}\ell\geq 0,
\end{displaymath}
i.e., the number of internal nodes at depth $\ell$ in the tree. Using Kraft's inequality (see, e.g., Cover and Thomas \cite{CoverThomas1991}), which states that for any binary tree,
\begin{displaymath}
\sum_{u\in\mathcal{L}}{\frac{1}{2^{\mathrm{d}(u)}}}\leq 1,
\end{displaymath}
we have, with $A(u)$ denoting the set of ancestors $u$ and $D(v)$ denoting the set of descendants of $v$, that, since $d(u)=\sum{\mathds{1}\{v\in A(u)\setminus\{u\}\}}$,
\begin{align*}
\mathbf{E}(T_{\epsilon})&=\sum_{u\in\mathcal{L}}\mathop{\sum_{v \in A(u)}}_{v\neq u}{\frac{1}{2^{\mathrm{d}(v)}}\frac{1}{2^{\mathrm{d}(u)-\mathrm{d}(v)}}}\\
&=\sum_{v\in\mathcal{I}}{\frac{1}{2^{\mathrm{d}(v)}}}\mathop{\sum_{u\in D(v)}}_{u\in\mathcal{L}}{\frac{1}{2^{\mathrm{d}(u)-\mathrm{d}(v)}}}\\
&\leq \sum_{v\in\mathcal{I}}{\frac{1}{2^{\mathrm{d}(v)}}}\phantom{1}\text{(by Kraft's inequality)}\\
&=\sum_{\ell=0}^{\infty}{\frac{N_{\ell}}{2^{\ell}}}.
\end{align*}
Observe that at depth $\ell$, all intervals associated with nodes are disjoint, and thus, since each interval node corresponds to an interval strictly larger than $2\epsilon$,
\begin{displaymath}
N_{\ell}<\frac{L}{2\epsilon}.
\end{displaymath}
Also, $N_{\ell}\leq 2^{\ell}$ because we have a binary tree. Hence,
\begin{displaymath}
N_{\ell}\leq \min\big\{\lfloor\rfrac{L}{2\epsilon}\rfloor,2^{\ell}\big\}.
\end{displaymath}
We deduce, using
\begin{displaymath}
\ell_{0}=\max\bigg\{0,\bigg\lceil\log_{2}\frac{L}{2\epsilon}\bigg\rceil\bigg\},
\end{displaymath}
that
\begin{align*}
\sum_{\ell=0}^{\infty}{\frac{N_{\ell}}{2^{\ell}}}&\leq\sum_{\ell=0}^{\ell_{0}}{1}+\sum_{\ell=\ell_{0}+1}^{\infty}{\bigg\lfloor\frac{L}{2\epsilon}\bigg\rfloor\frac{1}{2^{\ell}}}\\
&=\ell_{0}+1+\bigg\lfloor\frac{L}{2\epsilon}\bigg\rfloor\frac{1}{2^{\ell_0}}\\
&\leq\log_{2}^{+}\bigg(\frac{L}{2\epsilon}\bigg)+3.
\end{align*}
To see this, treat the cases $L<2\epsilon$ and $L\geq 2\epsilon$ separately.
\end{proof}

%%%%%%%%%%%%%%%%%%%%% bisection-inversion here ? algo II %%%%%%%%%%%%%%%%%%%%%%%%%%%%%%

\section{A rejection algorithm for densities with compact support}\label{sect:compactcase}

In this section, we assume that $f$ is Riemann-integrable and supported on $[0,1]^{d}$. Note that this is equivalent to the assumption that $f$ is almost-everywhere continuous, bounded, and supported on $[0,1]^{d}$. Our algorithm requires an oracle---a black box---that for any closed rectangle $R\subseteq\mathbb{R}^{d}$ gives
\begin{displaymath}
\sup_{x\in R}{f(x)}\,\text{ and }\inf_{x\in R}{f(x)}.
\end{displaymath}
If $R=\{x\}$, then that oracle coincides with a standard function evaluation. Without the possibility of computing infimum and supremum of the density $f$ over compact subintervals of the domain of $f$, sampling absolutely continuous distribution using the rejection method seems to be impossible in total generality. We will also track complexity in terms of the number of uses of the oracle before halting, and call it $T_{\epsilon}$. One use of the oracle reveals
\begin{displaymath}
C\isdef\sup_{x\in[0,1]^{d}}{f(x)},
\end{displaymath}
which is a finite number by assumption (Riemann-integrable functions are bounded by definition). At once, we have a simple bound for applying the rejection method: $f(x)\leq C$. Algorithm \ref{orig_vnabc} shows the original algorithm by von Neumann \cite{vN51} (see also Devroye \cite{Devroye1986}) of the standard rejection algorithm.

%\begin{flushright}
%\color{gray}{\textbf{The algorithm is at the beginning of the next page.}}\color{black}{}
%\end{flushright}
%\vfill
%\clearpage

\begin{algorithm}[h!]
\caption{Von Neumann's original rejection algorithm}\label{orig_vnabc}
\begin{algorithmic}[1]
\LOOP
\STATE Generate $X$ uniformly on $[0,1]^{d}$.
\STATE Generate $U$ uniformly on $[0,1]$.
\IF{$UC\leq f(X)$}
\RETURN $X$
\ENDIF
\ENDLOOP
\end{algorithmic}
\end{algorithm}

Since we cannot generate $X$ and $U$ with infinite precision, at least two modifications are needed. One modification is to take into account the precision $\epsilon$ desired for $X$, and another modification is to take into account that bits of $U$ are generated sequentially. Appendix \ref{app_b}, which is in this article for pedagogical purpose, gives more insights on wrong and naive modifications that someone could be tempted to do. We can consider the initial rectangle
\begin{displaymath}
R_{0}=[0,1]^{d}\times [0,C],
\end{displaymath}
and its $2^{d+1}$ child rectangles defined by the $2^{d+1}$ quadrants centered at the center $x_{0}$ of $R_{0}$:
\begin{displaymath}
x_{0}=\Big(\frac{1}{2},\frac{1}{2},\ldots,\frac{1}{2},\frac{C}{2}\Big).
\end{displaymath}
In the data structure literature, such a partition, when applied recursively, leads to a quadtree (see, e.g., Samet \cite{samet2006foundations}).
Any subsequent rectangle can be split again about its center point, and so forth, in the manner of infinite quadtree $Q$ on $R_{0}\subseteq \mathbb{R}^{d}$ as illustrated by Figures \ref{fig_2dim_decomp_for_quadtree} and \ref{fig_quadtree_decomp_example}.

\begin{flushright}
\color{gray}{\textbf{The figures are at the beginning of the next page.}}\color{black}{}
\end{flushright}
\vfill
\clearpage

\begin{figure}[h!]
\begin{centering}
\includegraphics[trim = 10mm 182mm 20mm 33mm]{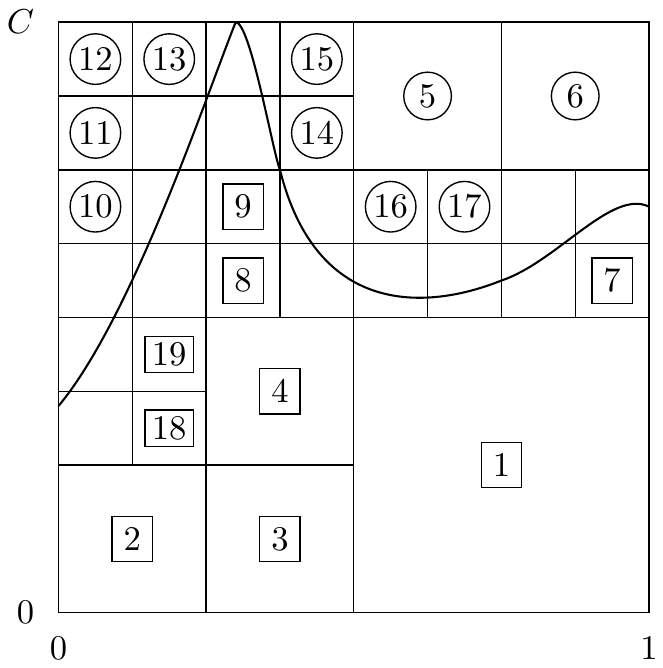}
\caption{Decomposition of $[0,1]\times [0,C]$ into its quadtree tree structure.}\label{fig_2dim_decomp_for_quadtree}
\vspace{5mm}
\includegraphics[trim = 40mm 212mm 20mm 33mm]{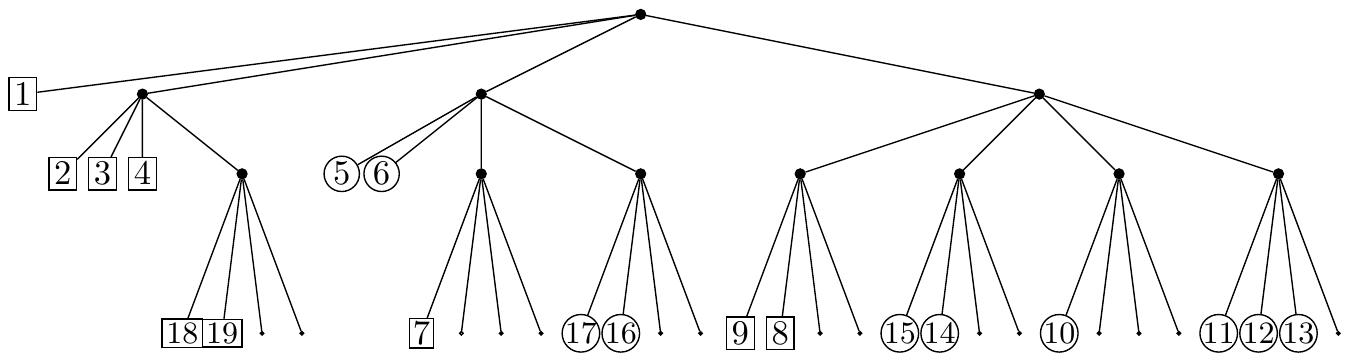}
\caption{Quadtree for decomposition on Figure \ref{fig_2dim_decomp_for_quadtree}.}\label{fig_quadtree_decomp_example}
\end{centering}
\end{figure}

In von Neumann's algorithm, deciding if $UC\leq f(X)$ for $(X,U)\in R_{0}$ is equivalent to finding a rectangle $R$ in the quadtree $Q$ with the property that either
\begin{align*}
&R\subseteq \big\{(x,y)\in R_{0}:\phantom{1} y\leq f(x)\big\}\phantom{12}\text{(we accept since $UC\leq f(X)$)}
\end{align*}
or
\begin{align*}
&R\subseteq \big\{(x,y)\in R_{0}:\phantom{1} y> f(x)\big\}\phantom{12}\text{(we reject since $UC> f(X)$)}.
\end{align*}
However, this must be done carefully, without overlapping rectangles. Thus, we must trim $Q$, and associate the exiting rectangles $R$ with the leaves. Thus,
\begin{align}
\big\{(x,y)\in R_{0}:\phantom{1} y\leq f(x)\big\}&=\bigcup\{R:\phantom{1}\text{$R$ is an accepting rectangle}\},\label{accept_rect}
\end{align}
and
\begin{align}
\big\{(x,y)\in R_{0}:\phantom{1} y> f(x)\big\}&=\bigcup\{R:\phantom{1}\text{$R$ is a rejecting rectangle}\}.\label{reject_rect}
\end{align}

Below, we will see that Riemann-integrability of $f$ suffices for this decomposition. When a rejecting rectangle is found, the procedure is repeated. When an accepting rectangle is found, say,
\begin{displaymath}
\prod_{i=1}^{d}{[a_i,b_i]}\times [\alpha,\beta],
\end{displaymath}
it suffices to generate $X_{\epsilon}\in\prod_{i=1}^{d}{[a_i,b_i]}$ such that
\begin{align}
\|X^{\star}-X_{\epsilon}\|\leq\epsilon\label{xstar_xepsilon}
\end{align}
where $X^{\star}$ is uniform on $\prod_{i=1}^{d}{[a_i,b_i]}$ and coupled with $X_{\epsilon}$. It is easy to see by the triangle inequality that $X_{\epsilon}$ is then coupled with $X$ having density $f$ such that
\begin{displaymath}
\|X-X_{\epsilon}\|\leq \epsilon.
\end{displaymath}
To achieve (\ref{xstar_xepsilon}), use bisection for each dimension separately. By Theorem \ref{thm_bisect}, the expected number of bits needed is bounded by
\begin{align*}
3+\sum_{i=1}^{d}{\log_{2}^{+}\bigg(\frac{b_i-a_i}{\epsilon}\bigg)}&=3+\mathop{\sum_{i=1}^{d}}_{b_i-a_i>\epsilon}{\log_{2}\bigg(\frac{b_i-a_i}{\epsilon}\bigg)}\\
&\leq 3+d\log_{2}\bigg(\frac{1}{\epsilon}\bigg)
\end{align*}
for $\epsilon\leq 1$. If $\epsilon>1$, then bisection requires no bit, so that we conclude that the expected number of bits does not exceed
\begin{displaymath}
3+d\log_{2}^{+}\bigg(\frac{1}{\epsilon}\bigg).
\end{displaymath}

We note that the checks
\begin{align}
&R\subseteq \{(x,y)\in\mathbb{R}^{d}\times\mathbb{R}:\phantom{1}f(x)\leq y\}\label{check_one}
\end{align}
and
\begin{align}
&R\subseteq \{(x,y)\in\mathbb{R}^{d}\times\mathbb{R}:\phantom{1}f(x) > y\}\label{check_two}
\end{align}
can be carried out using our oracle. Let

\begin{align*}
f^{+}&=\sup\big\{f(x):\phantom{1}(x,y)\in R\text{ for some $y$}\big\},\\
f^{-}&=\inf\big\{f(x):\phantom{1}(x,y)\in R\text{ for some $y$}\big\},\\
y^{-}&=\inf\{y:\phantom{1}(x,y)\in R\text{ for some $x$}\big\},\\
y^{+}&=\sup\{y:\phantom{1}(x,y)\in R\text{ for some $x$}\big\}.
\end{align*}
Then (\ref{check_one}) holds if $f^{+}\leq y^{-}$, and (\ref{check_two}) holds if $f^{-}\geq y^{+}$.

%\begin{flushright}
%\color{gray}{\textbf{The algorithm is at the beginning of the next page.}}\color{black}{}
%\end{flushright}
%\vfill
%\clearpage

\begin{algorithm}[h!]
\caption{Generation of an $\epsilon$-approximate random variable with density on $[0,1]^{d}$}\label{algo_reject_compact}
\begin{algorithmic}[1]
\STATE $R\leftarrow [0,1]^{d}\times\big[0,\mathop{\sup_{}}_{x\in[0,1]^{d}}{f(x)}\big]$\label{init_line_reject_algo}
\STATE $\texttt{Decision}\leftarrow \texttt{None}$
\REPEAT
\STATE Call the oracle that returns $\inf_{x\in R^{\star}} {f}(x)$ and $\sup_{x\in R^{\star}} {f}(x)$ which are in turn used by the following branching statement. \COMMENT{Here and below $R^{\star}$ denotes the projection of $R$ onto $\mathbb{R}^{d}$, i.e., $R^{\star}=\{x:\phantom{1}(x,y)\in R\text{ for some $y$}\}$.}
\IF {$R\subseteq \{(x,y)\in \mathbb{R}^{d}\times \mathbb{R}:\phantom{1}f(x)\leq y\}$}
\STATE $\texttt{Decision}\leftarrow \texttt{Accept}$
\ELSIF {$R\subseteq \{(x,y)\in \mathbb{R}^{d}\times \mathbb{R}:\phantom{1}f(x)> y\}$}
\STATE $\texttt{Decision}\leftarrow \texttt{Reject}$
\ELSE%[$\texttt{Decision}=\texttt{None}$]
\STATE $x^{\star}\leftarrow \text{center of }R$
\STATE Select one vertex $v$ of $R$ uniformly at random and replace $R$ by the rectangle with $v$ and $x^{\star}$ as opposing vertices.
\ENDIF
\UNTIL{$\texttt{Decision} \neq \texttt{None}$}
\IF{$\texttt{Decision}=\texttt{Reject}$}
\STATE Goto line (\ref{init_line_reject_algo}) \COMMENT{Restart the algorithm an average of $\sup_{x\in[0,1]^{d}}{f(x)}$ times.}
\ELSE%[$\texttt{Decision}=\texttt{Accept}$]
\STATE Use bisection to generate an $\epsilon$-approximation $X_{\epsilon}$ of a uniform variable in \mbox{$R^{\star}$}.\label{bisecto_sub_routine_line_reject_algo}
\RETURN $X_{\epsilon}$
\ENDIF
\end{algorithmic}
\end{algorithm}

\begin{theorem}\label{thm_quadtree}
Let $f$ be a Riemann-integrable density on $[0,1]^{d}$. Then the quadtree $Q$ partitions $\textstyle{R_{0}\isdef[0,1]^{d}\times [0,\sup{f}]}$ in a collection of leaf rectangles $R$ for which (\ref{accept_rect}) and (\ref{reject_rect}) hold. Thus, Algorithm \ref{algo_reject_compact} halts with probability one and delivers an $\epsilon$-approxiation $X_{\epsilon}$ of $X$, a random variable with density $f$.
\end{theorem}

\begin{proof}[Proof of Theorem \ref{thm_quadtree}]
Let $T$ be the number of iterations of the algorithm before halting. In other words, $T$ is the depth of the leaf rectangle $R$ reached by randomly walking down the quadtree. That walk costs $T(d+1)$ random bits. We show that
\begin{displaymath}
\lim_{k\to\infty}{\mathbf{P}\{T>k\}}=0,
\end{displaymath}
and thus, (\ref{accept_rect}) and (\ref{reject_rect}) must hold. In the partition of $R_{0}$ into $2^{(d+1)k}$ level-$k$ rectangles (each of Lebesgue measure $\squash{\frac{1}{2^{k}}}\squash{\frac{1}{2^{k}}}\cdots\squash{\frac{1}{2^{k}}}\squash{\frac{C}{2^{k}}}$), there are $N_{k}$ cells $R$---those for which we cannot decide---that are ``visited'' by $f$, i.e.\, for which
\begin{align*}
&\sup_{(x,y)\in R}{f(x)}\geq \inf\{y:\phantom{1}(x,y)\in R\text{ for some $x$}\}
\end{align*}
and
\begin{align*}
&\inf_{(x,y)\in R}{f(x)}\leq \sup\{y:\phantom{1}(x,y)\in R\text{ for some $x$}\}.
\end{align*}
Then
\begin{displaymath}
\mathbf{P}\{T>k\}=\frac{N_k}{2^{(d+1)k}}.
\end{displaymath}

For every rectangle $R$, let us define its projection, $R^{\star}$, onto $\mathbb{R}^{d}$, i.e.\,
\begin{displaymath}
R^{\star}=\{x:\phantom{1}(x,y)\in R\text{ for some $y$}\}.
\end{displaymath}

Then, for a fixed dimension, we can group the $2^{k}$ cells $R^{\star}$ with the same projection and verify that of these $2^{k}$ cells, at most
\begin{displaymath}
\bigg(\frac{\mathop{\sup_{}}_{x\in R^{\star}}{f(x)}-\mathop{\inf_{x\in R^{\star}}}_{}{f(x)}}{C}\bigg)2^{k}+2
\end{displaymath}
are visited by $f$. Since there are $2^{dk}$ rectangles $R^{\star}$, let $\mathcal{P}_{k}^{\star}$ be the collection of all projections $R^{\star}$, and then
\begin{align*}
N_{k}&\leq \sum_{R^{\star}\in\mathcal{P}_{k}^{\star}}{\Bigg(\bigg(\frac{\sup_{x\in R^{\star}}{f(x)}-\inf_{x\in R^{\star}}{f(x)}}{C}\bigg)2^{k}+2\Bigg)}\\
&=\frac{\big(I^{+}-I^{-}\big)}{C}2^{(d+1)k}+2\cdot2^{dk},
\end{align*}
where we use the Riemann approximations $I^{+}$ and $I^{-}$ of the integral of $f$:
\begin{align*}
I_{k}^{+}&=\sum_{R^{\star}\in\mathcal{P}_{k}^{\star}}{\Big(\sup_{x\in R^{\star}}{f(x)}\Big)\lambda(R^{\star})},\\
I_{k}^{-}&=\sum_{R^{\star}\in\mathcal{P}_{k}^{\star}}{\Big(\inf_{x\in R^{\star}}{f(x)}\Big)\lambda(R^{\star})},\\
\lambda(R^{\star})&=\text{Lebesgue measure of $R^{\star}$}=\frac{1}{2^{dk}}.
\end{align*}
Therefore,
\begin{displaymath}
\mathbf{P}\{T>k\}\leq \frac{2}{2^{k}}+\frac{I_{k}^{+}-I_{k}^{-}}{C}.
\end{displaymath}
Since $f$ is Riemann-integrable, this tends to $0$ as $k\to\infty$.

\end{proof}

We call $f$ a monotone density on $[0,1]^{d}$ if it decreases along at least one of the dimensions, i.e., there exists $i\in\{1,\ldots,d\}$ such that for all vectors \mbox{$(x_1,\ldots,x_i,\ldots,x_d)\in[0,1]^{d}$}, \mbox{$(x_1,\ldots,x'_i,\ldots,x_d)\in[0,1]^{d}$}, \mbox{$x_i\leq x'_i$}, then \mbox{$f(x_1,\ldots,x_i,\ldots,x_d)\geq f(x_1,\ldots,x'_i,\ldots,x_d)$}. As before, let $N_{k}$ be the number of cells at level $k$ that are visited by $f$. Then,
\begin{displaymath}
N_{k}\leq 2\cdot2^{k}\cdot2^{(d-1)k}
\end{displaymath}
because the domain of $f$ is divided into $2^{dk}$ cells and the $2^{k}$ cells along the $i^{\text{th}}$ dimension give rise to a walk. This walk along the $i^{\text{th}}$ is at most of length $2\cdot2^{k}$ as illustrated on Figure \ref{fig_mono_curve}.
%\begin{flushright}
%\color{gray}{\textbf{The algorithm is at the beginning of the next page.}}\color{black}{}
%\end{flushright}
%\vfill
%\clearpage
\begin{figure}[h!]
\begin{centering}
\includegraphics[trim = 28mm 147mm 14mm 35mm]{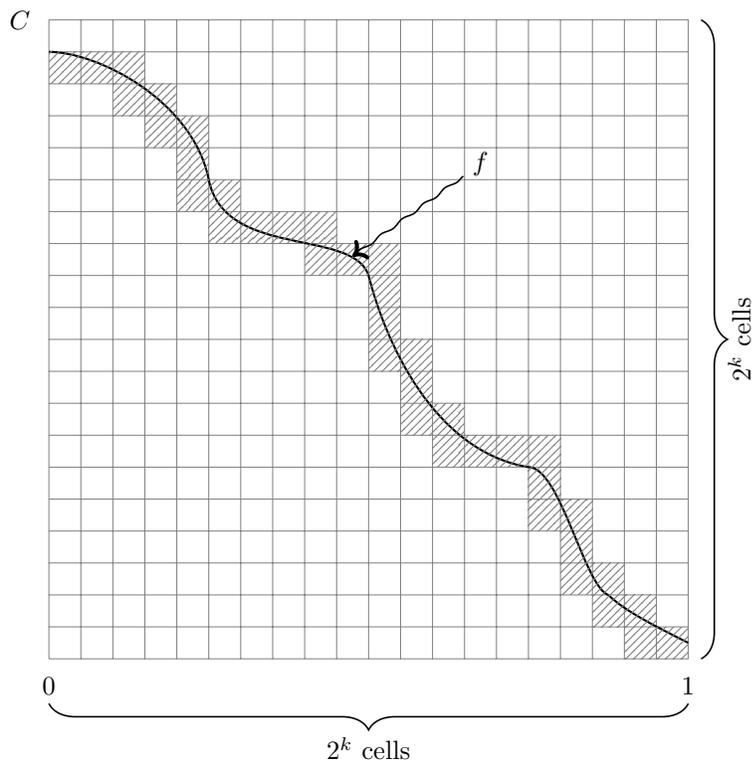}
\caption{The number of cells visited by the monotone curve $f$ is at most $2\cdot2^{k}$ cells.}\label{fig_mono_curve}
\end{centering}
\end{figure}

We have
\begin{displaymath}
\mathbf{P}\{T>k\}=\frac{N_k}{2^{(d+1)k}}\leq\frac{2}{2^{k}},\phantom{1}k\geq 0,
\end{displaymath}
and thus,
\begin{displaymath}
\mathbf{E}(T)=\sum_{k=0}^{\infty}{\mathbf{P}\{T>k\}}\leq 4.
\end{displaymath}
In other words, for monotone densities, the inner loop of the algorithm has a uniform performance guarantee.

For a complete analysis of algorithm $A$, we need to consider the total number $N$ of trials before deciding. The number of uses of the oracle is given by
\begin{displaymath}
\sum_{i=1}^{N}{T_{i}},
\end{displaymath}
where $T_i$ is the number of iterations in the $i$-th trial---these $T_i$'s are i.i.d. The number of random bits used is
\begin{displaymath}
(d+1)\sum_{i=1}^{N}{T_i}
\end{displaymath}
during the first phase (the decision phase), and is bounded by
\begin{displaymath}
3+d\log_{2}^{+}\bigg(\frac{1}{2\epsilon}\bigg)
\end{displaymath}
in the second phase (the bisection phase). Since $\mathbf{E}(N)=C=\sup_{x\in[0,1]^{d}}{f(x)}$, we see that the expected number of uses of the oracle is
\begin{displaymath}
C\mathbf{E}(T)
\end{displaymath}
and that the expected number of random bits required is bounded from above by
\begin{displaymath}
(d+1)C\mathbf{E}(T)+3+d\log_{2}^{+}\bigg(\frac{1}{2\epsilon}\bigg).
\end{displaymath}
As noted earlier, for any coordinate-wise monotone density on $[0,1]^{d}$,
\begin{displaymath}
\mathbf{E}(T)\leq 4.
\end{displaymath}
However, for general Riemann-integrable densities we cannot insure that $\mathbf{E}(T)$ converges. We will address that point below.

\begin{theorem}\label{thm_compact_case}
Let $f$ be a Riemann-integrable density on $[0,1]^{d}$, with \mbox{$C\isdef\sup_{x\in[0,1]^{d}}{f(x)}$}.
\begin{enumerate}
\item If $f$ is monotone in at least one coordinate, then the expected number of uses of the oracle is not more than
\begin{displaymath}
4C,
\end{displaymath}
and the expected number of bits needed to generate an $\epsilon$-approximation $X_{\epsilon}$ is not more than
\begin{displaymath}
4C(d+1)+3+d\log_{2}^{+}\bigg(\frac{1}{2\epsilon}\bigg).
\end{displaymath}
\item If $I_{k}^{+}$ and $I_{k}^{-}$ are the Riemann approximation of $\int{f}$ for regular grids of size $2^{dk}$ i.e., each coordinate is split into $2^{k}$ equal intervals, then the expected number of uses of the oracle is not more than
\begin{displaymath}
4C+\mathcal{A}(f),
\end{displaymath}
where
\begin{displaymath}
\mathcal{A}(f)\isdef \sum_{k=0}^{\infty}{\big(I_{k}^{+}-I_{k}^{-}\big)},
\end{displaymath}
and the expected number of random bits used by Algorithm \ref{algo_reject_compact} is not more than
\begin{displaymath}
4C(d+1)+(d+1)\mathcal{A}(f)+3+d\log_{2}^{+}\bigg(\frac{1}{2\epsilon}\bigg).
\end{displaymath}
\end{enumerate}
\end{theorem}

\begin{proof}[Proof of Theorem \ref{thm_compact_case}]
Just recall the estimates of $\mathbf{E}(T)$ obtained above and recall the upper bound $\mathbf{P}\{T>k\}$ in terms of $I_{k}^{+}-I_{k}^{-}$.
\end{proof}

\begin{remark}
Part (2) of the Theorem \ref{thm_compact_case} is only useful if $\mathcal{A}(f)<\infty$. For most densities, $\mathcal{A}(f)<\infty$, as can be seen from this simple sufficient condition. Let the modulus of continuity be
\begin{displaymath}
\omega(\delta)=\sup_{\|x-y\|\leq\delta}{|f(x)-f(y)|},\phantom{1}\delta>0.
\end{displaymath}
We have $\mathcal{A}(f)<\infty$ if
\begin{align*}
\hspace{0.5in}\sum_{k=0}^{\infty}{\omega\bigg(\frac{\sqrt{d}}{2^{k}}\bigg)}&<\infty
\end{align*}
because
\begin{align*}
I_{k}^{+}-I_{k}^{-}&=\frac{1}{2^{dk}}\sum_{R^{\star}\in\mathcal{P}_{k}^{\star}}{\bigg(\sup_{x\in R^{\star}}{f}(x)-\inf_{x\in R^{\star}}{f}(x)\bigg)}\\
&\leq \frac{1}{2^{dk}}\sum_{R^{\star}\in\mathcal{P}_{k}^{\star}}{\sup_{x,x'\in R^{\star}}|f(x)-f(x')|}\\
&\leq \omega\bigg(\frac{\sqrt{d}}{2^{k}}\bigg).
\end{align*}
It suffices that as $\delta\downarrow 0$, $\omega(\delta)=O\big(\rfrac{1}{\log^{1+\alpha}(\delta^{-1})}\big)$ or $\omega(\delta)=O(\delta^{\alpha})$ for some $\alpha>0$.
\end{remark}

\begin{remark}
The number of bits used by our algorithm behaves as $\textstyle{d\log_{2}^{+}\big(\frac{1}{\epsilon}\big)+O(1)}$ as $\epsilon\downarrow 0$, and thus matches the lower bound mentioned earlier.
\end{remark}

\section{A rejection algorithm for densities with non-compact support}\label{sect:noncompacta1}

In general, we use von Neumann's rejection method when we know a density $g$ for which random variate generation is ``easy'', and can verify that
\begin{displaymath}
\sup_{x\in\mathbb{R}}{\frac{f(x)}{g(x)}}= C<\infty,
\end{displaymath}
where $C$ is a known constant:

%\begin{flushright}
%\color{gray}{\textbf{The algorithm is at the beginning of the next page.}}\color{black}{}
%\end{flushright}
%\vfill
%\clearpage
\begin{algorithm}
\caption{General rejection algorithm}\label{gen_algo}
\begin{algorithmic}
\LOOP
\STATE Generate $X$ with density $g$
\STATE Generate $U$ uniformly on $[0,1]$
\IF{$Cg(X)U<f(X)$}
\RETURN $X$ \COMMENT{Exit: $X$ is accepted}
\ENDIF
\ENDLOOP
\end{algorithmic}
\end{algorithm}

The expected number of iterations of Algorithm \ref{gen_algo} is $C$. We offer a generalization of Algorithm \ref{gen_algo} for this situation under a certain number of assumptions. Assume for now that $d=1$, and that we can compute both $G$ and $G^{-1}$, where $G$ is the c.d.f.\ for $g$. Assume furthermore that we have an oracle for
\begin{displaymath}
\sup_{x\in R}{\frac{f(x)}{g(x)}}\phantom{1}\text{and}\phantom{1}\inf_{x\in R}{\frac{f(x)}{g(x)}}
\end{displaymath}
for all intervals $R$ of $\mathbb{R}$. One may be able to work things out with oracles for $\sup{f}$, $\inf{f}$, $\sup{g}$, and $\inf{g}$ as well but we opt to take the more convenient approach.

Define $C=\sup_{x\in\mathbb{R}}{\frac{f(x)}{g(x)}}$, which is known thanks to our oracle. The goal is to decompose
\begin{displaymath}
\big\{(x,y):\phantom{1}y\leq f(x)\big\},
\end{displaymath}
as before, into regions for which random variate generation is ``easy''. For a bounded density on $[0,1]$, we are content with the rectangles. This can be mimicked by transforming the $x$-axis with
\begin{displaymath}
x\mapsto G(x)
\end{displaymath}
since $G$ is monotone and continuous. Using this transformation, we note that if $X$ has density $g$, then $G(X)$ is uniform on $[0,1]$. Furthermore, note that if $u=G(x)$, then
\begin{displaymath}
\frac{f\circ G^{-1}(u)}{g\circ G^{-1}(u)}=\frac{f(x)}{g(x)}\isdef \tilde{f}(u),
\end{displaymath}
where $\tilde{f}$ is a density on $[0,1]$ on which we can use our $\sup$-$\inf$ oracle. Since $\tilde{f}\leq C$, we can use the quadtree method of Algorithm \ref{algo_reject_compact} to select a rectangle $R_i$ with probability $\lambda(R_i)$ in the decomposition
\begin{displaymath}
\big\{(u,v):\phantom{1}v\leq \tilde{f}(u),\phantom{1}u\in[0,1]\big\}=\bigcup_{i\in\mathbb{N}}{\big\{R_{i}:\phantom{1}\text{$R_{i}$ is an accepting rectangle}\big\}}.
\end{displaymath}
This decomposition is valid, and the procedure halts with probability one, if $\tilde{f}$ is Riemann-integrable. Put differently, it works if
\begin{displaymath}
\tilde{f}(u)=\frac{f\circ G^{-1}(u)}{g\circ G^{-1}(u)},\phantom{1}0<u<1,
\end{displaymath}
is Riemann-integrable. The expected number of bits required in the decision phase of the algorithm (selecting a leaf of the quadtree) is bounded as in Theorem \ref{thm_compact_case}, when applied to $\tilde{f}$. It is bounded by
\begin{displaymath}
\Delta \isdef 2\bigg(4C+\sum_{k=0}^{\infty}{\big(I_{k}^{+}-I_{k}^{-}\big)}\bigg),
\end{displaymath}
where $I_{k}^{+}$ and $I_{k}^{-}$ are the Riemann approximation of $\int_{0}^{1}{\tilde{f}(u)\mathrm{d}u}$ for a grid of $2^{k}$ equal intervals that partition $[0,1]$.

We only need to analyze the second phase---the method to generate an $\epsilon$-approximation once a rectangle $R$ has been selected, i.e., the method that starts by accepting rectangle \mbox{$R=[u_1,u_2]\times[v_1,v_2]\subseteq [0,1]\times [0,C]$} as illustrated by Figure \ref{fig_bisect} and outputs $X_{\epsilon}$ such that $\textstyle{\|X_{\epsilon}-G^{-1}(U)\|\leq\epsilon}$ where $X$, $X_{\epsilon}$ are coupled and $(U,V)$ is uniform in $R$. This can be achieved by bisection, noting that $G^{-1}(U)$ has distribution function $G$ restricted to $[G^{-1}(u_1),G^{-1}(u_2)]$.
\begin{flushright}
\color{gray}{\textbf{A figure illustrating the principle is at the beginning of the next page.}}\color{black}{}
\end{flushright}
\vfill
\clearpage
\begin{figure}[h!]
\begin{centering}
\includegraphics[trim = 36mm 170mm 18mm 40mm]{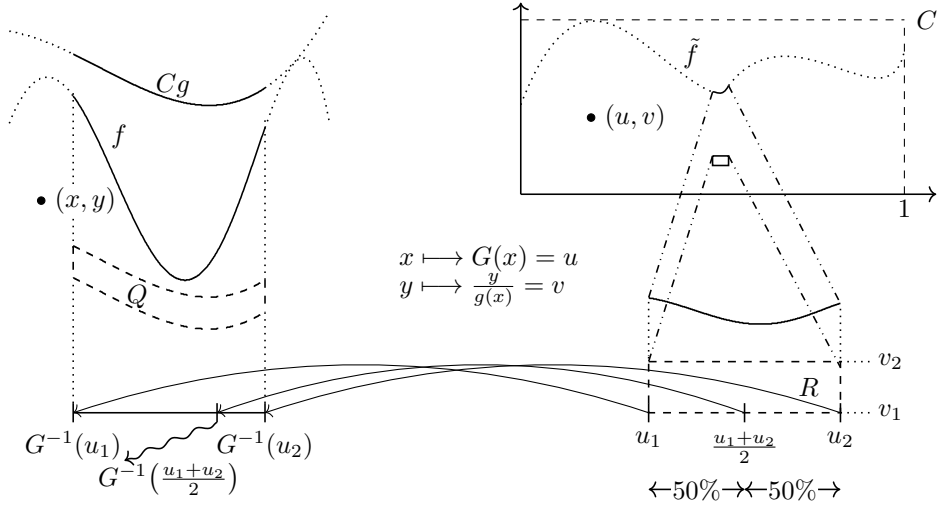}
\caption{An accepting uniform random rectangle $R$ and the bisection of its back-transformation $Q$: $(u,v)\in R$ if and only if $(x,y)\in Q$.}\label{fig_bisect}
\end{centering}
\end{figure}

\begin{remark}
Note that with $x_1=G^{-1}(u_1)$ and $x_2=G^{-1}(u_2)$, we have
\begin{displaymath}
\lambda(R)=(u_2-u_1)(v_2-v_1)=\int_{u_1}^{u_2}{(v_2-v_1)\mathrm{d}u}=\int_{x_1}^{x_2}{(v_2-v_1)(g(x)\mathrm{d}x)}=\lambda(Q).
\end{displaymath}
Also if $(u,v)$ is such that $v<\tilde{f}(u)$, then the corresponding $(x,y)$ point is such that
\begin{displaymath}
y=vg(x)<\tilde{f}(u)g(x)=\bigg(\frac{f(x)}{g(x)}\bigg)g(x)=f(x),
\end{displaymath}
and similarly for $v>\tilde{f}(u)$.
\end{remark}

The inversion Algorithm \ref{algo:inv_bisecto} of Devroye and Gravel \cite{DevGra2015paper1}, which is extension of a similar method for the discrete case first proposed by Han and Hoshi \cite{HanHoshi1997} is summarized as follows:

\begin{flushright}
\color{gray}{\textbf{The algorithm is at the beginning of the next page.}}\color{black}{}
\end{flushright}
\vfill
\clearpage
\begin{algorithm}[h!]
\caption{Inversion/Bisection}\label{algo:inv_bisecto}
\begin{algorithmic}[1]
\REQUIRE $u_1$, $u_2$ such that $u_2>u_1$ \COMMENT{$u_2-u_1$ is the width of an accepting rectangle.}
\STATE $x_1\leftarrow G^{-1}(u_1)$
\STATE $x_2\leftarrow G^{-1}(u_2)$
\LOOP
\IF{$|x_2-x_1|\leq 2\epsilon$}
\STATE $X_{\epsilon}\leftarrow\frac{x_1+x_2}{2}$ \COMMENT{Midpoint}
\RETURN $X_{\epsilon}$
\ELSE%[$|x_2-x_1|> 2\epsilon$]
\STATE $\gamma\leftarrow \frac{u_1+u_2}{2}$
\STATE $B\leftarrow\texttt{random unbiased bit}$ \COMMENT{To choose a random side.}
\IF{$B=0$}
\STATE $u_2\leftarrow \gamma$
\STATE $x_2\leftarrow G^{-1}(u_2)$
\ELSE%[$B=1$]
\STATE $u_1\leftarrow \gamma$
\STATE $x_1\leftarrow G^{-1}(u_1)$
\ENDIF
\ENDIF
\ENDLOOP
\end{algorithmic}
\end{algorithm}

This algorithm picks a uniform subinterval and, if permitted to run forever, would produce a random variable with distribution function $G$ restricted to \mbox{$[G^{-1}(u_1),G^{-1}(u_2)]$} as is illustrated in Figure \ref{fig_bisect}. So, for random variate generation, we only replace line (\ref{bisecto_sub_routine_line_reject_algo}) of Algorithm \ref{algo_reject_compact} by Algorithm \ref{algo:inv_bisecto} where $[u_1,u_2]=R^{*}$, and note that elsewhere in Algorithm \ref{algo_reject_compact}, $f$ must be replaced by $\tilde{f}$.

\begin{theorem}\label{thm_gen_case_I}
The expected number of bits used by Algorithm \ref{algo:inv_bisecto} is not more than
\begin{displaymath}
3+\sum_{j\in\mathbb{Z}}{F(I_j)\log_{2}\bigg(\frac{1}{F(I_j)}\bigg)},
\end{displaymath}
where $I_j=[2\epsilon j, 2\epsilon(j+1))$, and $F\big([a,b)\big)\isdef F(b)-F(a)$. In particular, if that sum is finite for $\epsilon=1$ and if $\textstyle{\int{f\log_{2}\big(1\slash f\big)}>-\infty}$, then, as $\epsilon\downarrow 0$, the expected number of bits does not exceed
\begin{displaymath}
\log_{2}\bigg(\frac{1}{\epsilon}\bigg)+\int{f\log_{2}\bigg(\frac{1}{\epsilon}\bigg)}+5+o(1).
\end{displaymath}
\end{theorem}

\begin{remark}
Theorem \ref{thm_gen_case_I} establishes that Algorithm \ref{algo:inv_bisecto} is optimal to within an additive constant. In particular, its main term,  $\textstyle{\log_{2}\big(1\slash \epsilon\big)}$, and second term, the differential entropy $\textstyle{\int{f\log_{2}\big(1\slash f\big)}}$, match the lower bound.
\end{remark}

\begin{remark}
The expected number of bits required in the decision phase of the algorithm, $\Delta$, is finite under smoothness conditions on $\tilde{f}$. It depends also on $C$, but clearly not on $\epsilon$.
\end{remark}

\begin{proof}[Proof of Theorem \ref{thm_gen_case_I}]
Let us denote an accepting rectangle $R_i$ and its projection by $R_{i}^{\star}$. So, if $R_{i}^{\star}=[u_i,v_i]$, then $R_i=[u_i,v_i]\times [\alpha_i,\alpha_i+Cq_i]$, where $0\leq \alpha_i\leq \alpha_i+Cq_i\leq C$, $q_i\in[0,1]$. The probability mass of $R_i$ is $p_i\isdef (v_i-u_i)C q_i$.

By the mapping $G^{-1}$, $R_i$ gets mapped to a contiguous region $Q_i$, of projection \mbox{$Q_i^{\star}\isdef [a_i,b_i]$}, with
\begin{displaymath}
a_i=G^{-1}(u_i)\text{ , }b_i=G^{-1}(v_i),
\end{displaymath}
and thus,
\begin{displaymath}
v_i-u_i = \int_{a_i}^{b_i}{g}=G(Q_i^{\star})=\frac{p_i}{C q_i}.
\end{displaymath}
Here we use the notation $G([a_i,b_i])=G(b_i)-G(a_i)$. We also note that for all $x$,
\begin{displaymath}
\sum_{i:\, x\in Q_{i}^{\star}}{C q_i g(x)}=f(x).
\end{displaymath}
Define a regular $2\epsilon$-grid on $\mathbb{R}$ with intervals $I=[2\epsilon j,2\epsilon(j+1))$ for all $j\in\mathbb{Z}$.

If we exit with rectangle $R_i$, then the bisection phase of the algorithm takes an expected number of bits bounded by
\begin{displaymath}
3+\sum_{j\in\mathbb{Z}}{\xi_{ji}\log_{2}\bigg(\frac{1}{\xi_{ji}}\bigg)},
\end{displaymath}
where
\begin{displaymath}
\xi_{ji}=\frac{G\big(I_j\cap Q_{i}^{\star}\big)}{G(Q_{i}^{\star})}\text{ , }j\in\mathbb{Z}
\end{displaymath}
is a probability vector in $j$. This result is due to the observation that the bisection method is equivalent to the algorithm analyzed by Devroye and Gravel \cite{DevGra2015paper1} and Gravel \cite{gravelphdthesis} in the context of the discrete distribution algorithm by Han and Hoshi \cite{HanHoshi1997}. Thus, the expected number of bits, averaged over all $R_i$, is not more than
\begin{align}
&3+\sum_{i\in\mathbb{Z}}{p_i\sum_{j\in\mathbb{Z}}{\xi_{ji}\log_{2}\bigg(\frac{1}{\xi_{ji}}\bigg)}}\label{ikuiku_no1}.
\end{align}
By the concavity of $u\log_{2}\big(1\slash u\big)$ in $u$, we have by Jensen's inequality that (\ref{ikuiku_no1}) is not more than
\begin{displaymath}
3+\sum_{j\in\mathbb{Z}}\Bigg(\sum_{i\in\mathbb{Z}}{p_i\xi_{ji}}\Bigg)\log_{2}\Bigg(\frac{1}{\sum_{i\in\mathbb{Z}}{p_i\xi_{ji}}}\Bigg).
\end{displaymath}
But
\begin{align*}
\sum_{i\in\mathbb{Z}}{p_i\xi_{ji}}&=\sum_{i\in\mathbb{Z}}{p_i\frac{G\big(I_j\cap Q_{i}^{\star}\big)}{G(Q_{i}^{\star})}}\\
&=\sum_{i\in\mathbb{Z}}{C q_i G(I_j \cap Q_{i}^{\star})}\\
&=\sum_{i\in\mathbb{Z}}{C q_i \int_{I_j}{\mathds{1}_{\{x\in Q_i^{\star}\}}g(x)\mathrm{d}x}}\\
&=\int_{I_j}{\Bigg(\sum_{i\in\mathbb{Z}}{C q_i \mathds{1}_{\{x\in Q_{i}^{\star}\}}}\Bigg)g(x)\mathrm{d}x}\\
&=\int_{I_j}{f(x)\mathrm{d}x}\\
&\isdef F(I_{j}),
\end{align*}
where $F$ is the distribution function of $f$, and $F(I_j)=F(b_j)-F(a_j)$. Thus the expected number of bits in the bisection phase does not exceed
\begin{align}
&3+\sum_{j\in\mathbb{Z}}{F(I_j)\log_{2}\bigg(\frac{1}{F(I_j)}\bigg)}\label{ikuiku_no2}.
\end{align}
In the last term, we recognize the entropy defined by the probability vector $\big(F(I_j)\big)_{j\in\mathbb{Z}}$.

A theorem due to Csisz\'{a}r \cite{Csiszar1961} established that if $\big(F(I_j)\big)_{j\in\mathbb{Z}}$ has a finite entropy for some $\epsilon>0$, and $\int{f\log_{2}\big(1\slash f\big)}>-\infty$, then as $\epsilon\downarrow 0$,
\begin{displaymath}
(\ref{ikuiku_no2})\leq \int{f\log_{2}\frac{1}{f}}+\log_{2}\frac{1}{\epsilon}+5+o(1).
\end{displaymath}
The ``5'' can be replaced by ``3'' if in addition $f$ is bounded and decreasing on its support, $[0,\infty)$ (see Gravel and Devroye \cite{DevGra2015paper1}).

\end{proof}

\section{Conclusion and outlook}
The extension of our results to dimensions greater than one
for densities with unbounded support should pose no big problems.
With the oracles introduced in our modification of von Neumann's
method, we believe that it is impossible to design a rejection algorithm
for densities that are not Riemann-integrable, so the question of
the design of a universally valid rejection algorithm under the random
bit model remains open.

\section*{Acknowledgment}

The authors would like to thank all three referees for their feedback. They are invited to Luc's house for drinks if they wish to come out of
the closet. Luc Devroye's research was supported by an NSERC Discovery Grant. Claude Gravel thanks Gilles Brassard.

\bibliographystyle{siam}  %http://www.cs.stir.ac.uk/~kjt/software/latex/showbst.html

\newcommand{\SortNoop}[1]{}

\appendix

\section{Riemann integrability and the sup/inf oracle}\label{app_a}

In this section, we contruct a family of densities that are not Riemann-integrable for which the oracle is useless. Let $\delta\in[0,1\slash 3)$ be a parameter, and let $I_{\delta}\subseteq [0,1]$ be a Cantor-like set constructed below. Setting
\begin{displaymath}
f_{\delta}(x)=\frac{1}{\lambda(I_{\delta})}\mathds{1}_{\{x\in I_{\delta}\}},
\end{displaymath}
where $\lambda$ is the Lebesgue measure, we have (see below) $\lambda(I_{\delta})=(1-3\delta)\slash (1-2\delta)$. The Lebesgue measure of the set of discontinuities is $1-\lambda(I_{\delta})$ and is zero only if $\delta=0$. The case of $\delta=0$ corresponds to the usual uniform density on $[0,1]$ that is Riemann-integrable. All cases of $\delta\in(0,1\slash 3)$ are Lebesgue integrable but not Riemann-integrable because the set of discontinuities is non-zero and yet possess a cumulative distribution function. For every $\epsilon>0$ and every $x\in I_{\delta}$ we have
\begin{align*}
&\inf_{[x-\epsilon,x+\epsilon]}{f(x)}=0,\\
\text{and}&\\
&\sup_{[x-\epsilon,x+\epsilon]}{f(x)}=\frac{1-2\delta}{1-3\delta}.
\end{align*}
Since these boundaries are invariant under changes of $\epsilon$, Algorithm \ref{algo_reject_compact}, when used for rejection, say, from a uniform density, has an infinite loop with positive probability. It is, therefore, essential that Riemann-integrable densities are considered as that the supremum and infimum returned by the oracle over given small intervals converge to each other as intervals shrink.

For the construction of $I_{\delta}$, we recursively remove middle \emph{open} subintervals of geometrically decreasing sizes. Let $I_{j,k}\subset [0,1]$ for $j\in\mathbb{N}\setminus\{0\}$ and $k=0,\ldots{},2^{j}-1$ be the $2^{j}$ \emph{closed} subintervals that are \emph{left} once the middle parts are removed from the previous subintervals $I_{j-1,k}$ with $k=0$ corresponding to the leftmost subinterval and so on. Initially, $I_{0,0}=[0,1]$. For all $j\in\mathbb{N}\setminus\{0\}$ and $k\in\{0,\ldots{},2^{j}-1\}$, the length of a removed middle part is $\delta^{j}$. For all $j\in\mathbb{N}\setminus\{0\}$, the total length \emph{not} removed at the $j$-th step is $2^{j}\lambda(I_{j,0})$ because the subintervals are of the same length. Let $I_{\delta}$ be the limiting subset of $[0,1]$ that is left, i.e.,
\begin{displaymath}
I_{\delta}=\bigcap_{j=1}^{\infty}{\bigcup_{k=0}^{2^{j}-1}{I_{j,k}}}.
\end{displaymath}
We compute $\lambda(I_{\delta})$ as follows:
\begin{align*}
\lambda(I)&=\lim_{j\to\infty}{\sum_{k=0}^{2^{j}-1}{\lambda(I_{j,k})}}\phantom{1}\text{(by the definition of $I_{\delta}$)}\\
&=\lim_{j\to\infty}{2^{j}\lambda(I_{j,0})},\\
\lambda(I_{1,0})&=\frac{1}{2}-\frac{\delta}{2},\\
\lambda(I_{j,0})&=\frac{1}{2}\lambda(I_{j-1,0})-\frac{\delta^{j}}{2}\\
&=\frac{1}{2^{j}}-\frac{\delta}{2^{j}}-\frac{\delta^{2}}{2^{j-1}}-\ldots-\frac{\delta^{j}}{2}\phantom{1}\text{for $j\geq 1$,}
\end{align*}
and therefore
\begin{displaymath}
2^{j}\lambda(I_{j,0})=1-\delta\sum_{i=0}^{j-1}{(2\delta)^{i}}=1-\delta\bigg(\frac{(2\delta)^{j}-1}{2\delta-1}\bigg),
\end{displaymath}
so that
\begin{displaymath}
\lambda(I_{\delta})=\frac{1-3\delta}{1-2\delta}.
\end{displaymath}

\section{A naive modification to the general rejection method that is incorrect}\label{app_b}

A trivial, but incorrect, modification to Algorithm \ref{orig_vnabc} would be:
%\begin{flushright}
%\color{gray}{\textbf{The doomed algorithm is at the beginning of the next page.}}\color{black}{}
%\end{flushright}
%\vfill
%\clearpage
\begin{algorithm}[h!]
\begin{algorithmic}[1]
\LOOP
\STATE By bisection, generate $X_{\epsilon}$ on $[0,1]^{d}$ such that $X_{\epsilon}$ is an $\epsilon$-approximation of $X$.
\STATE Generate $U$ uniformly on $[0,1]$.
\STATE Decide ``$U\leq \frac{1}{C}f(X_{\epsilon})$'' by using two bits in expected value.\label{doomed_algo_line1}
\STATE If the condition from the previous line is satisfied, then return $X_{\epsilon}$.
\ENDLOOP
\end{algorithmic}
\end{algorithm}

This attempt leads to failure. Let $A$ be the support of $X_{\epsilon}$, which is necessarily contained in a fixed countable subset of $[0,1]^{d}$. Then given any Riemann-integrable density $g$, take a finite subset $A^{\star}$ of $A$, and modify $g$ on $A^{\star}$ by setting
\begin{equation*}
f(x) = \left\{
\begin{array}{rl}
g(x) & \text{if } x \notin A^{\star},\\
C & \text{if } x \in A^{\star}.
\end{array} \right.
\end{equation*}
We have that $f$ is still a Riemann-integrable density bounded by $C$ (and its set of discontinuities is of measure $0$ because $A^{\star}$ is countable) but since $f(X_{\epsilon})=C$ if $X_{\epsilon}\in A^{\star}$, we accept all $X_{\epsilon}\in A^{\star}$, regardless of the density $g$ we started with. Since $\mathbf{P}\{X_{\epsilon}\in A^{\star}\}>0$, we make an error with positive probability.

If we set
\begin{equation*}
f(x) = \left\{
\begin{array}{rl}
g(x) & \text{if } x \notin A,\\
0 & \text{if } x \in A,
\end{array} \right.
\end{equation*}
then $g$ is no longer Riemann-integrable, and in that case, $f(X_{\epsilon})=0$ with probability one, and therefore, the algorithm loops forever, regardless of the choice of $g$.

This simple example shows the necessity of the oracle and of the condition of Riemann-integrability.

\end{document}